\documentclass[letterpaper, 10 pt, conference]{ieeeconf}  

\IEEEoverridecommandlockouts              
\makeatletter

\let\proof\@undefined
\let\endproof\@undefined

\makeatother
\usepackage{graphicx}
\usepackage{caption}
\usepackage{subcaption}
\usepackage{enumerate}
\usepackage{hyperref}
\usepackage{upgreek}
\usepackage{array}
\usepackage{multirow}
\usepackage{booktabs}
\usepackage{amsmath} 
\usepackage{amssymb}  
\usepackage{amsthm}  
\usepackage{bm}
\usepackage{lipsum}
\usepackage[noend, ruled, linesnumbered]{algorithm2e}
\usepackage{color}
\usepackage{cite}
\usepackage{marginnote,zref-savepos}

\newtheorem{proposition}{Proposition}

\newtheorem{theorem}{Theorem}

\newcolumntype{M}[1]{>{\centering\arraybackslash}m{#1}}

\newcounter{labelnote}
\makeatletter
\let\oldmarginnote\marginnote
\renewcommand*{\marginnote}[1]{%
 \begingroup\strut
  \stepcounter{labelnote}\zsaveposx {marginnote-\thelabelnote}
     \ifnum 0\zposx{marginnote-\thelabelnote}<200000
      \reversemarginpar
      \oldmarginnote{\color{blue}#1}%
     \else
      \normalmarginpar
      \oldmarginnote{\color{blue}#1}%
     \fi
 \endgroup%
}
\makeatother

\title{\LARGE \bf
Distributionally Robust Differential Dynamic Programming with Wasserstein Distance\thanks{This work was supported in part by the National Research Foundation of Korea under MSIT2020R1C1C1009766,  the Information and Communications Technology Planning and Evaluation under Grants MSIT2020-0-00857, MSIT2022-0-00124, MSIT2022-0-00480, and  Samsung Electronics.}
}

\author{Astghik Hakobyan \and Insoon Yang
\thanks{A. Hakobyan, and I. Yang are with the Department of Electrical and Computer Engineering and ASRI, Seoul National University, Seoul, 08826, Korea {\tt\small \{astghikhakobyan, insoonyang\}@snu.ac.kr}}%
}

\begin{document}

\maketitle
\thispagestyle{empty}
\pagestyle{empty}


\begin{abstract}
Differential dynamic programming (DDP) is a popular technique for solving nonlinear optimal control problems with locally quadratic approximations. However, existing DDP methods are not designed for stochastic systems with unknown disturbance distributions. To address this limitation, we propose a novel DDP method that approximately solves the Wasserstein distributionally robust control (WDRC) problem, where the true disturbance distribution is unknown but a disturbance sample dataset is given.  Our approach aims to develop a practical and computationally efficient DDP solution.
To achieve this, we use the Kantrovich duality principle to decompose the  value function in a novel way and
derive closed-form expressions of the distributionally robust control and worst-case  distribution policies to be used in each iteration of our DDP algorithm.
This characterization makes our method tractable and scalable without the need for numerically solving any minimax optimization problems.
The superior out-of-sample performance and scalability of our algorithm are demonstrated through  kinematic car navigation and coupled oscillator problems. 
\end{abstract}


\section{Introduction}\label{sec:intro}

Nonlinear optimal control problems are difficult to solve exactly, particularly when the state space dimension is high. 
Differential dynamic programming (DDP) alleviates this issue using locally-quadratic approximations of the system dynamics and cost function~\cite{liao1991convergence, tassa2014control, pavlov2021interior, jallet2022implicit, so2022maximum, roulet2022iterative}.
It efficiently computes an approximate solution with superior scalability compared to the standard dynamic programming (DP) approach.  
However,  it is generally challenging to apply DDP to systems with  random disturbances without any means to counteract them.

Although various  works have extended  DDP  to handle stochastic systems, existing methods often rely on either the ground truth or potentially inaccurate approximate probability distributions of disturbances. For example, the DDP algorithms introduced in~\cite{todorov2005generalized,theodorou2010stochastic, pan2015data, pan2018efficient} either consider Gaussian multiplicative noise or model the uncertain system dynamics as Gaussian processes. 
 Another line of research is devoted to the minimax formulation of the DDP problem (e.g.,~\cite{sun2018min,morimoto2003minimax}), where the optimal control problem is solved in the face of the worst-case disturbances. However, such methods often lead to overly conservative solutions.

To address the limitations of stochastic DDP methods and handle systems with unknown disturbance distributions, we propose a novel approach inspired by distributionally robust control (DRC).
The objective of DRC is to design control policies that maximize the worst-case performance over a set of candidate distributions without assuming a specific distribution of disturbances. 
Several techniques have been proposed for hedging against distributional uncertainties in DRC problems, including moment-based and statistical distance-based approaches~\cite{Yang2018, Schuurmans2020, Mark2021,  coulson2021distributionally,Zolanvari2021, hakobyan2021wasserstein, Dixit2022,Micheli2022}. 
While moment-based approaches rely on accurate moment estimates and may not effectively capture the full distributional information about the uncertainties, distance-based methods consider distributions that are close to a given nominal one in terms of a statistical distance measure. Many recent works have focused on Wasserstein DRC (WDRC)~\cite{yang2020wasserstein, zhong2021data, kordabad2022safe,Hakobyan2022conf, kim2022minimax}, where the ambiguity set is designed as a statistical ball with the distance between two distributions measured by the Wasserstein metric. The Wasserstein ambiguity set has  salient features, including a  finite-sample performance guarantee and the ability to avoid pathological solutions to distributionally robust optimization (DRO) problems~\cite{mohajerin2018data, gao2022distributionally,boskos2020data}.

Despite numerous attempts, existing WDRC methods still face challenges in terms of tractability and scalability. For instance, the DP-based approach introduced in~\cite{yang2020wasserstein}  for solving the WDRC problem results in a semi-infinite program, requiring computationally expensive state-space discretization or sampling. To overcome this limitation, both~\cite{yang2020wasserstein} and~\cite{kim2022minimax} propose a relaxation technique with a penalty on the Wasserstein distance, which leads to an explicit solution in the linear-quadratic (LQ) setting. While these works focus on the theoretical analysis of the obtained policies, our approach provides a practical and computationally efficient algorithm for solving the nonlinear WDRC problem.

In particular, a novel DDP method is developed through a locally quadratic approximation of a nonlinear WDRC problem, where the true disturbance distribution is unknown but a disturbance sample is given.
By construction, the proposed distributionally robust DDP (DR-DDP) algorithm provides  control policies that are robust against inevitable inaccuracies in empirical distributions of the disturbance.  
 To make the method tractable, we first approximate the WDRC problem with its penalty version and then apply the Kantorovich duality principle.   We show that the proposed approximation provides a suboptimal solution to the original WDRC problem. The value function is then decomposed in a novel way that enables us to derive computationally tractable  and efficient backward and forward passes. This allows us to obtain closed-form expressions for the distributionally robust control and worst-case distribution policies in each iteration of the DR-DDP algorithm. By avoiding the need for numerically solving minimax optimization problems, our approach makes the algorithm not only tractable but also scalable.
 The scalability of our DDP method is a remarkable advantage because the computational complexity of the standard DP algorithm in~\cite{yang2020wasserstein} for nonlinear WDRC increases exponentially with the dimension of the state space. 
 The experiment results on kinematic car navigation and coupled oscillator problems indicate that our algorithm outperforms existing methods in terms of out-of-sample performance and provides scalable solutions for high-dimensional nonlinear optimal control problems.

\section{Preliminaries}\label{sec:prel}

In this section, we introduce the WDRC problem  used in our development of the DR-DDP algorithm in Section~\ref{sec:DRDDP}.

\subsection{Distributionally Robust Control}
Consider the following discrete-time stochastic system:
\begin{equation}\label{dyn}
x_{t+1} = f(x_t,u_t, w_t),
\end{equation}
where $x_t\in\mathbb{R}^{n_x}$ and $u_t\in\mathbb{R}^{n_u}$ are the system states and control inputs, respectively. Here, $w_t\in\mathbb{R}^{n_w}$ is a random disturbance with an unknown (true) distribution $\mathbb{Q}_t^\mathrm{true} \in \mathcal{P}(\mathbb{R}^{n_w})$, where $\mathcal{P}(\mathbb{R}^{n_w})$ is the family of all Borel probability measures supported on $\mathbb{R}^{n_w}$. The nonlinear function $f:\mathbb{R}^{n_x} \times \mathbb{R}^{n_u} \times \mathbb{R}^{n_w} \to \mathbb{R}^{n_x}$ is assumed to be twice continuously differentiable.

In practice, it is restrictive to assume that the true probability distribution $\mathbb{Q}_t^\mathrm{true}$ is known. Instead, we are often given a sample dataset $\mathcal{D}_t := \{\hat{w}_t^{(1)}, \hat{w}_t^{(2)}, \dots, \hat{w}_t^{(N)}\}$ drawn from the true distribution, which can be used to construct an empirical estimate about the distribution of $w_t$ as
\[
\mathbb{Q}_t:= \frac{1}{N} \sum_{i=1}^{N}\delta_{\hat{w}_t^{(i)}},
\]
 where $\delta_{\hat{w}_t^{(i)}}$ denotes the Dirac measure concentrated at $\hat{w}_t^{(i)}$. It is well-known that as $N\to\infty$, the empirical distribution asymptotically converges to the true distribution. However, if an inaccurate empirical estimate is used in the controller design, the resulting control performance will deteriorate due to a mismatch between the true and empirical distributions.

 To hedge against such distributional uncertainties, we adopt a game-theoretic approach and consider a two-player zero-sum game in which Player I is the controller and Player II is a hypothetical adversary. Let $\pi:= (\pi_0, \ldots, \pi_{T-1})$ denote the control policy, where $\pi_t$ maps the state $x_t$ to a control input $u_t$. 
 The adversary player selects a policy $\gamma :=(\gamma_0, \dots, \gamma_{T-1})$, where $\gamma_t$ maps the current state to a probability distribution $\mathbb{P}_t$ chosen from an \emph{ambiguity set} $\mathbb{D}_t \subset \mathcal{P}(\mathbb{R}^{n_w})$. The ambiguity set is a family of distributions that possess certain properties to be described.  
 
 Throughout this paper, our goal is to design an optimal finite-horizon controller with the following cost functional:
\[
J(\pi, \gamma):= \mathbb{E}^{\pi,\gamma}\big [\ell_f(x_T) + \sum_{t=0}^{T-1}\ell(x_t, u_t)\big],
\]
where $\ell:\mathbb{R}^{n_x}\times \mathbb{R}^{n_u}\to \mathbb{R}$ and $\ell_f:\mathbb{R}^{n_x}\to \mathbb{R}$ are the twice continuously differentiable running and terminal costs, respectively, and $T$ is the time horizon. In our problem, the controller seeks a policy $\pi^*$ minimizing the cost function, while the adversary aims to find a policy $\gamma^*$ to maximize the same cost, which can be obtained by solving the following DRC problem:
\begin{equation}\label{minimax}
\min_{\pi\in\Pi} \max_{\gamma\in\Gamma_\mathbb{D}} J(\pi, \gamma),
\end{equation}
where $\Pi:=\{\pi \mid \pi_t(x_t) = u_t \in\mathbb{R}^{n_u}, \; \forall t\}$ and $\Gamma_\mathbb{D}:=\{\gamma \mid \gamma_t(x_t) = \mathbb{P}_t\in \mathbb{D}_t, \; \forall t\}$ are the sets of admissible control and distribution policies, respectively.

\subsection{Wasserstein Ambiguity Set}

In problem~\eqref{minimax}, the adversary player is restricted to select a distribution from the ambiguity set $\mathbb{D}_t$, which determines the characteristics of the worst-case distribution. Therefore, it is necessary to design the ambiguity set to appropriately characterize distributional errors.
Motivated by its advantages mentioned in Section~\ref{sec:intro},   we use the Wasserstein  ambiguity set constructed around the given empirical distribution.
The Wasserstein metric of order $p$ between two distributions  $\mathbb{P}$ and $\mathbb{Q}$ supported on $\mathcal{W} \subseteq \mathbb{R}^n$ represents the minimum cost of redistributing mass from one distribution to another using a small non-uniform perturbation and is defined as
\[
\begin{split}
W_p(\mathbb{P}, \mathbb{Q}) := \inf_{\tau \in \mathcal{P}(\mathcal{W}^2)} \bigg\{  \Big(\int_{\mathcal{W}^2} & \|x-y\|^p \, \mathrm{d}\tau(x,y)  \Big)^{1/p} 
\\
&\big| \, \Pi^1 \tau = \mathbb{P}, \Pi^2 \tau = \mathbb{Q} \bigg\},
\end{split}
\]
where  $\tau$ is the \emph{transport plan} with $\Pi^i\tau$ denoting its $i$th marginal distribution, and $\|\cdot\|$ is a norm on $\mathbb{R}^n$ which quantifies the transportation cost.

In this work, we consider the Wasserstein metric of order $p=2$ with the transportation cost represented by the standard Euclidean norm. We design the ambiguity set as follows:
\begin{equation}\label{amb_set}
\mathbb{D}_t:=\{ \mathbb{P}_t\in\mathcal{P}(\mathbb{R}^{n_w}) \mid W_2(\mathbb{P}_t, \mathbb{Q}_t)\leq \theta\},
\end{equation}
where $\theta >0$ determines the size of $\mathbb{D}_t$. The ambiguity set~\eqref{amb_set} is a statistical ball centered at the empirical distribution $\mathbb{Q}_t$ and contains all distributions whose Wasserstein distance from the empirical distribution is no greater than radius $\theta$.

\section{Distributionally Robust Differential Dynamic Programming}\label{sec:DRDDP}

In this section, 
 we present our main result, called DR-DDP, which efficiently finds an approximate solution to the WDRC problem.
 Our method exploits the Kantorovich duality principle to decompose the value function in a novel way
  and
 devise a computationally tractable algorithm.

\subsection{Approximation with Wasserstein Penalty}

In~\cite{kim2022minimax}, the tractability and effectiveness of a penalty version of the WDRC problem are studied. 
Motivated by this work, 
we begin our reformulations by replacing the Wasserstein ambiguity set constraint with a penalty term in the cost function as follows:
\[
J_\lambda(\pi, \gamma):= \mathbb{E}^{\pi,\gamma}\big[ \ell_f(x_T)+ \sum_{t=0}^{T-1} \ell(x_t, u_t) 
- \lambda W_2(\mathbb{P}_t,\mathbb{Q}_t)^2\big],
\]
where $\lambda >0$ is the penalty parameter adjusting the conservativeness of the controller. 

Then, the following minimax control problem approximates the original WDRC problem~\eqref{minimax}:
\begin{equation}\label{minimax_penalty}
\min_{\pi\in\Pi}\max_{\gamma\in\Gamma} J_\lambda(\pi, \gamma),
\end{equation}
where the adversary player selects policies from $\Gamma:=\{\gamma:=(\gamma_0,\dots, \gamma_{T-1}) \mid \gamma_t(x_t) = \mathbb{P}_t\in\mathcal{P}(\mathbb{R}^{n_w})\}$. Note that the adversary is not restricted to select distributions from the ambiguity set. Instead, we penalize large deviations from the empirical distribution via the penalty term, thus limiting the freedom of the adversary player.

We demonstrate in the following proposition that the cost incurred by an arbitrary policy $\pi\in\Pi$ under the worst-case distributions within the Wasserstein ambiguity set has a guaranteed cost property with respect to the worst-case penalized cost. Hence, the penalty problem~\eqref{minimax_penalty} is a reasonable approximation as it yields a suboptimal solution to the WDRC problem~\eqref{minimax}.

\begin{proposition}\label{prop:gc}
Given $\lambda > 0$, 
let $\pi \in \Pi$ be any arbitrary policy.
Then, the cost incurred by $\pi$ under the worst-case distribution policy in $\Gamma_{\mathbb{D}}$ is upper-bounded  as follows:
\begin{equation}\label{guarantee}
 \sup_{\gamma \in \Gamma_{\mathbb{D}}} J(\pi, \gamma) \leq  \lambda T \theta^2 +  \sup_{\gamma \in \Gamma} J_\lambda (\pi, \gamma).
\end{equation}
\end{proposition}

Its proof can be found in Appendix~\ref{app:gc}.
The guaranteed cost property indicates the role of the penalty parameter $\lambda$ in adjusting the robustness of the control policy, thereby providing a guideline on its selection. Specifically, the penalty parameter can be chosen to yield the least upper bound in~\eqref{guarantee} under the given control policy.

To formalize our algorithm, we recursively define the optimal value function for   problem~\eqref{minimax_penalty} as follows:
\begin{equation*}
\begin{split}
V_t(\bm{x}) := \inf_{\pi\in\Pi}\sup_{\gamma \in \Gamma}  \mathbb{E}^{\pi,\gamma}\bigg[ & \ell_f(x_T)+ \sum_{s=t}^{T-1} \ell(x_s, u_s) 
\\
&- \lambda W_2(\mathbb{P}_s,\mathbb{Q}_s)^2 \mid x_t = \bm{x}\bigg]
\end{split}
\end{equation*}
for $t=T-1,\dots, 0$, with the terminal condition $V_T(\bm{x})= \ell_f(\bm{x})$. Then, the DP principle yields
\begin{equation}\label{vf}
\begin{split}
V_t(\bm{x}) & = \inf_{\bm{u}\in\mathbb{R}^{n_u}}\sup_{\mathbb{P} \in \mathcal{P}(\mathbb{R}^{n_w})}  \ell (\bm{x},\bm{u}) \\
&+ \mathbb{E}^{w \sim \mathbb{P}}\bigg[ V_{t+1}(f(\bm{x},\bm{u},w))
- \lambda W_2(\mathbb{P},\mathbb{Q}_t)^2\bigg]
\end{split}
\end{equation}
with the optimal cost given by
\[
J_\lambda^* := \inf_{\pi\in\Pi}\sup_{\gamma\in\Gamma} J_\lambda(\pi,\gamma) = V_0(x_0).
\]

Unfortunately, the standard procedure for DDP cannot be applied to the value function~\eqref{vf} as it constitutes an infinite-dimensional optimization problem over $\mathcal{P}(\mathbb{R}^{n_w})$. For tractability, we employ a modern DRO technique based on the Kantorovich duality principle~\cite{yang2020wasserstein, sinha2017certifying}
and reformulate the value function as follows.

\begin{proposition}\label{prop:vf}
Suppose that for each $(\bm{x}, \bm{u})\in\mathbb{R}^{n_x}\times \mathbb{R}^{n_u}$, the value function is measurable and that the outer minimization problem in~\eqref{vf} has an optimal solution. Then, for any $\lambda > 0$, we have that
\begin{equation}\label{vf_kant}
\begin{split}
 V_t  (\bm{x})  &= \inf_{\bm{u}\in\mathbb{R}^{n_u}}  \ell (\bm{x},\bm{u})  +  \mathbb{E}^{\hat{w}_t\sim \mathbb{Q}_t}\Big[\\
&\sup_{\bm{w} \in \mathbb{R}^{n_w}}  V_{t+1}(f(\bm{x},\bm{u},\bm{w}))
- \lambda \| \hat{w}_t - \bm{w}\|^2\Big],
\end{split}
\end{equation}
for all $\bm{x}\in\mathbb{R}^{n_x}$.
\end{proposition}

Its proof can be found in Appendix~\ref{app:vf}.
While previous works (e.g.,~\cite{yang2020wasserstein}) use similar approaches to reformulate and analyze the solution to the WDRC problem, our focus is on designing a practical and efficient method for obtaining tractable solutions.  For that, we let
\[
Q_t^{(i)}(\bm{x},\bm{u},\bm{w}) :=  \ell(\bm{x},\bm{u}) + V_{t+1}(f(\bm{x},\bm{u},\bm{w}))
- \lambda \| \hat{w}_t^{(i)} - \bm{w}\|^2
\]
denote the state-action-disturbance value function or the Q-function for each sample index $i=1,\dots,N$ and 
\[
Q_t^{*,(i)}(\bm{x}, \bm{u}) = \sup_{\bm{w} \in \mathbb{R}^{n_w}} Q_t^{(i)}(\bm{x},\bm{u},\bm{w})
\]
denote the corresponding ``worst-case'' state-action value function. Then, we obtain that
\begin{equation}\label{vf_Q}
V_t(\bm{x}) = \inf_{\bm{u}\in\mathbb{R}^{n_u}} \frac{1}{N}\sum_{i=1}^{N} Q_t^{*,(i)}(\bm{x},\bm{u}).
\end{equation}
It is worth emphasizing that  the Kantorovich duality principle enables us to obtain
this novel decomposition of the value function, which 
can be used to design a computationally tractable 
DR-DDP solution 
in the following subsection.

\subsection{Solution via DDP}\label{sec:sol_ddp}

In each iteration of the original DDP algorithm, a backward pass is performed on the current estimate of the state and control trajectories, called the \emph{nominal trajectories}, followed by a forward pass. In the backward pass, the cost function and the system dynamics are quadratically approximated around the nominal trajectories to update the policy, while in the forward pass, the nominal trajectories are recomputed by executing the latest policy to the system. We adopt this technique for our problem and derive the backward and forward passes for the value function~\eqref{vf_kant}.  The proposed DR-DDP method is
presented in Algorithm~\ref{alg:dr_ddp}.

\subsubsection{Backward Pass}\label{sec:backward_pass}

In each backward pass, we are given nominal state, control input, and disturbance trajectories $\bm{\bar{x}}_\mathrm{nom} = (\bm{\bar{x}}_0, \dots, \bm{\bar{x}}_T)$, $\bm{\bar{u}}_\mathrm{nom} = (\bm{\bar{u}}_0, \dots, \bm{\bar{u}}_{T-1})$ and $\bm{\bar{w}}_\mathrm{nom} = (\bm{\bar{w}}_0, \dots, \bm{\bar{w}}_{T-1})$, respectively. 
For quadratic approximations, DDP considers the following deviations of the system state, control input, and disturbance, i.e., $\delta x_t := x_t - \bm{\bar{x}}_t$, $\delta u_t:= u_t - \bm{\bar{u}}_t$, $\delta w_t:= w_t - \bm{\bar{w}}_t$.

We first consider the following second-order approximation of $V_{t+1}(x_{t+1})$:
\begin{equation}\label{vf_approx}
  \bm{V}_{t+1} + V_{t+1,x}^\top \delta x_{t+1} + \frac{1}{2} \delta x_{t+1}^\top V_{t+1,xx}\delta x_{t+1},
\end{equation}
for some $(\bm{V}_{t+1}, V_{t+1,x}, V_{t+1,xx})\in \mathbb{R} \times \mathbb{R}^{n_x} \times \mathbb{R}^{n_x\times n_x}$ to be determined.\footnote{If $V_{t+1}$ is twice differentiable, the parameters $(\bm{V}_{t+1}, V_{t+1,x}, V_{t+1,xx})$ can be simply determined using the second-order Taylor expansion. } 
Let $\hat{Q}_t^{(i)}$ be an approximate Q-function, defined by 
replacing $V_{t+1}$ in the definition of $Q_t^{(i)}$
with the approximate value function~\eqref{vf_approx}.
Then, $\hat{Q}_t^{(i)}( x_t, u_t,  w_t)$ is twice differentiable and its second-order Taylor expansion is given by
\begin{equation}\label{approx_Q}
\begin{split}
  \bm{Q}_t^{(i)} + \delta Q_t^{(i)}(\delta x_t, \delta u_t, \delta w_t),
\end{split}
\end{equation}
where 
\[
\begin{split}
\delta Q_t^{(i)}(\delta x_t, \delta u_t, \delta w_t) = Q_{t,x}^\top \delta x_t &+ Q_{t,u}^\top \delta u_t + {Q_{t,w}^{(i)}}^\top\delta w_t \\
&+ \frac{1}{2} \Delta Q_t(\delta x_t, \delta u_t, \delta w_t)
\end{split}
\]
with 
\[
\Delta Q_t( \delta x, \delta u, \delta w):=\begin{bmatrix}\delta x \\ \delta u\\ \delta w\end{bmatrix}^\top \begin{bmatrix} Q_{t,xx} & Q_{t,xu} & Q_{t,xw}\\ Q_{t,xu}^\top & Q_{t,uu} & Q_{t,uw} \\ Q_{t,xw}^\top & Q_{t,uw}^\top & Q_{t,ww} \end{bmatrix} \begin{bmatrix}\delta x \\ \delta u \\ \delta w\end{bmatrix}
\]
and
\begin{equation*}\label{q_update}
\left \{
\begin{array}{l}
\bm{Q}_t^{(i)} = \ell(\bm{\bar{x}}_t, \bm{\bar{u}}_t) + \bm{V}_{t+1}- \lambda\|\bm{\bar{w}}_t - \hat{w}_t^{(i)}\|^2\\
Q_{t,xx} = \ell_{t,xx} + f_{t,x}^\top V_{t+1,xx} f_{t,x} + V_{t+1,x}^\top f_{t,xx}\\
Q_{t,uu} = \ell_{t,uu} + f_{t,u}^\top V_{t+1,xx} f_{t,u}+ V_{t+1,x}^\top f_{t,uu}\\
Q_{t,ww} = f_{t,w}^\top V_{t+1,xx} f_{t,w} - 2\lambda I+ V_{t+1,x}^\top f_{t,ww}\\
Q_{t,xu} = \ell_{t,xu} + f_{t,x}^\top V_{t+1,xx} f_{t,u}\\
Q_{t,xw} = f_{t,x}^\top V_{t+1,xx} f_{t,w}, \quad Q_{t,uw} = f_{t,u}^\top V_{t+1,xx} f_{t,w}\\
Q_{t,x}  =  \ell_{t,x} + f_{t,x}^\top V_{t+1,x}, \quad Q_{t,u} = \ell_{t,u} +  f_{t,u}^\top V_{t+1,x}\\
Q_{t,w}^{(i)}  = f_{t,w}^\top V_{t+1,x}  - 2\lambda (\bm{\bar{w}}_t - \hat{w}_t^{(i)}).
\end{array}
\right.
\end{equation*}
Here, $f_{t,\cdot}$ and $\ell_{t,\cdot}$ denote the partial derivatives of $f$ and $\ell$ evaluated at $(\bar{\bm{x}}_t, \bar{\bm{u}}_t, \bar{\bm{w}}_t)$. 

Let
$\hat{\bar{w}}_t := \mathbb{E}^{\hat{w}_t\sim \mathbb{Q}_t}[\hat{w}_t]$ and
$\hat{\Sigma}_t := \mathbb{E}^{\hat{w}_t\sim \mathbb{Q}_t}[(\hat{w}_t - \hat{\bar{w}}_t)(\hat{w}_t - \hat{\bar{w}}_t)^\top]$
denote the empirical mean vector and covariance matrix of disturbance $w_t$, respectively.
The above approximation transforms the problem~\eqref{vf_Q} into a quadratic form similar to that addressed in \cite{kim2022minimax}. This approximation enables us to explicitly solve the problem with respect to $\delta u_t$ and $\delta w_t$, as presented in the following theorem.

\begin{theorem}\label{thm:backward}
Let $Q_{t,ww} \prec 0$ and $\ell_{t,uu} \succ 0$. Suppose the value function at time $t+1$ is approximated as~\eqref{vf_approx}. Then, the outer minimization problem in~\eqref{vf_Q} with $Q_{t}^{(i)}(x_t, u_t, w_t)$ replaced by the approximation~\eqref{approx_Q} has the following unique minimizer:
\begin{equation}\label{opt_cont}
\delta u_t^* = K_t \delta x_t + k_t,
\end{equation}
where
\begin{equation}\label{control}
\begin{split}
    K_t & = -\tilde{Q}_t ( Q_{t,xu}^\top - Q_{t,uw} Q_{t,ww}^{-1} Q_{t,xw}^\top )\\
    k_t & = -\tilde{Q}_t (  Q_{t,u} - Q_{t,uw}Q_{t,ww}^{-1} \bar{Q}_{t,w})
\end{split}
\end{equation}
with $\tilde{Q}_t:=(Q_{t,uu} - Q_{t,uw} Q_{t,ww}^{-1} Q_{t,uw}^\top)^{-1}$ and $\bar{Q}_{t,w} := f_{t,w}^\top V_{t+1,x} - 2\lambda (\bm{\bar{w}}_t - \hat{\bar{w}}_t)$.
 
Moreover, for each $i=1,\dots, N$, the maximization problem in~\eqref{vf_Q} with $Q_{t}^{(i)}(x_t, u_t, w_t)$ replaced by the approximation~\eqref{approx_Q} has the following unique solution:
\begin{equation}\label{wc_dist}
\delta w_t^{*,(i)} = H_t \delta x_t + h_t^{(i)} ,
\end{equation}
where 
\begin{equation}\label{wc_dist_para}
\begin{split}
    H_t &=  -Q_{t,ww}^{-1} [Q_{t,uw}^\top K_t + Q_{t,xw}^\top] \\ h_t^{(i)} &= - Q_{t,ww}^{-1}[Q_{t,uw}^\top k_t + Q_{t,w}^{(i)}].
\end{split}
\end{equation}
\end{theorem}

\begin{proof}  Let $\delta w_t^{(i)}:= w_t^{(i)} - \bm{\bar{w}}_t$. Evaluating the approximate Q-function~\eqref{approx_Q} for $\delta w_t^{(i)}$, we see that it is strictly concave in $\delta w_t^{(i)}$ as $Q_{t,ww} \prec 0$. Then, the first-order optimality condition yields the following unique maximizer:
\begin{equation}\label{wc_dist_1}
    \delta w_t^{*,(i)} = -Q_{t,ww}^{-1} (Q_{t,xw}^\top \delta x_t + Q_{t,uw}^\top \delta u_t + Q_{t,w}^{(i)}).
\end{equation}
Replacing $Q_{t}^{(i)}(x_t, u_t, w_t)$ with the approximation~\eqref{approx_Q}, the objective function in~\eqref{vf_Q} is quadratically approximated as
\[
\begin{split}
  \frac{1}{N}\sum_{i=1}^{N}\Big[  \bm{Q}_t^{(i)}  + \delta Q_t^{(i)} & (\delta x_t, \delta u_t, \delta w_t^{*,(i)})\Big]\\
 = \, &\bar{\bm{Q}}_t + Q_{t,x}^\top \delta x_t + Q_{t,u}^\top \delta u_t + \bar{Q}_{t,w}^\top \overline{\delta w}_t^* \\
 &+ \frac{1}{2} \Delta Q_t(\delta x_t, \delta u_t, \overline{\delta w_t}^*),
\end{split}
\]
where
\[
\begin{split}
\overline{\delta w}_t^* : =  \, &\frac{1}{N}\sum_{i=1}^N \delta w_t^{*,(i)} \\
= \, & - Q_{t,ww}^{-1}(Q_{t,xw}^\top \delta x_t + Q_{t,uw}^\top \delta u_t + \bar{Q}_{t,w})
\end{split}
\]
and
\[
\begin{split}
\bar{\bm{Q}}_t:= \ell_t(\bm{\bar{x}}_t, \bm{\bar{u}}_t) +  \bm{V}_{t+1} & - \lambda \|\bm{\bar{w}}_t - \hat{\bar{w}}_t\|^2 
-\lambda \mathrm{Tr}[\hat{\Sigma}_t] \\
&- 2\lambda^2 \mathrm{Tr}[Q_{t,ww}^{-1} {\hat{\Sigma}_t}].
\end{split}
\]
To minimize this approximated objective function with respect to $\delta u_t$,  the following first-order optimality condition can be used:
\[
\begin{split}
0 =\, & Q_{t,u} + Q_{t,uu} \delta u_t + Q_{t,xu}^\top \delta x_t + Q_{t,uw}\overline{\delta w}_t^* \\
&+ \frac{\partial \overline{\delta w}_t^*}{\partial \delta u_t}^\top (\bar{Q}_{t,w}  + Q_{t,xw}^\top \delta x_t + Q_{t,uw}^\top \delta u_t + Q_{t,ww}\overline{\delta w}_t^*).
\end{split}
\]
By the strong convexity of the quadratic approximation, its minimizer  is uniquely given by
\[
\begin{split}
\delta u_t^* =\, &  - \tilde{Q}_t \big( Q_{t,u} -  Q_{t,uw}Q_{t,ww}^{-1} \bar{Q}_{t,w}\\
&+[Q_{t,xu}^\top - Q_{t,uw} Q_{t,ww}^{-1} Q_{t,xw}^\top] \delta x_t \big),
\end{split}
\]
which is equivalent to~\eqref{opt_cont}. By substituting $\delta u_t^*$ into~\eqref{wc_dist_1}, we obtain the maximizer defined in~\eqref{wc_dist}.
\end{proof}

Theorem~\ref{thm:backward} provides the remarkable advantage that a DR-DDP  policy pair $(\bar{\pi}^*, \bar{\gamma}^*)$ is constructed in the following closed-form without numerically solving any infinite-dimensional minimax optimization problems:
\begin{subequations}
\begin{align}
\bar{\pi}_t^*(x_t) &= \bar{\bm{u}}_t + K_t (x_t - \bar{\bm{x}}_t) + k_t \label{aff_cont_pol}\\
\bar{\gamma}_t^*(x_t) &= \frac{1}{N}\sum_{i=1}^{N} \delta_{(\bar{\bm{w}}_t + h_t^{(i)} + H_t (x_t - \bm{\bar{x}}_t))}.\label{wc_dist_}
\end{align}
\end{subequations}
As a result of the backward pass, we also obtain the following  equations for updating the parameters of  the approximate value function~\eqref{vf_approx}:
\begin{equation}\label{quad_approx}
\begin{split}
\bm{V}_t =\, & \bar{\bm{Q}}_t + Q_{t,u}^\top k_t+\bar{Q}_{t,w}^\top h_t  \\
&+ \frac{1}{2}k_t^\top Q_{t,uu} k_t + \frac{1}{2}h_t^\top Q_{t,ww} h_t + k_t^\top Q_{t,uw} h_t\\
V_{t,x} = \,& Q_{t,x} + Q_{t,xu} k_t + K_t^\top  (Q_{t,u} + Q_{t,uu} k_t + Q_{uw} h_t) \\
&+Q_{xw} h_t+ H_t^\top (\bar{Q}_{t,w} + Q_{t,ww} h_t + Q_{t,uw}^\top k_t) \\
V_{t,xx} =\, & Q_{t,xx} + K_t^\top Q_{t,uu} K_t + H_t^\top Q_{t,ww} H_t + 2 Q_{t,xu} K_t \\
&+ 2K_t^\top Q_{t,uw} H_t + 2 Q_{t,xw} H_t,
\end{split}
\end{equation}
where $h_t:= \frac{1}{N} \sum_{i=1}^{N} h_t^{(i)}$.
%

In practice, it is not common to assume that control inputs are unrestricted. Often, the control inputs are limited to some box constraints $\overline{\mathbf{u}} \leq u_t \leq \underline{\mathbf{u}}$. Taking into account such control limits requires a careful design of the backward pass as it is required to minimize the approximate state-action value function subject to the constraints. To find a closed-form solution to the constrained problem for all $\delta x_t$, we use the projected Newton-based approach proposed in~\cite{tassa2014control}, where the control gains are found by solving a quadratic program.

In the next step, the nominal trajectories have to be reconstructed using the DR-DDP policy pair $(\bar{\pi}^*, \bar{\gamma}^*)$ to update the quadratically approximated models, which is performed during the forward pass introduced in what follows.

\subsubsection{Forward Pass}\label{sec:forward_pass}
In the original DDP algorithm, the forward pass is performed by executing the control policy to the system. However, due to the disturbance term in the system dynamics and lack of knowledge about its true distribution, it is not trivial to perform forward rollouts for the ambiguous stochastic  system~\eqref{dyn}. Instead, we choose to execute the control and distribution policy pair $(\bar{\pi}^*, \bar{\gamma}^*)$  in the following manner. 
First, using~\eqref{aff_cont_pol} and \eqref{wc_dist_}, we construct a control input $u_t = \bm{\bar{u}}_t + \alpha k_t + K_t (x_t - \bm{\bar{x}}_t)$ and sample a disturbance  realization as $w_t \sim  \frac{1}{N}\sum_{i=1}^{N} \delta_{\bm{\bar{w}}_t + \alpha h_t^{(i)} + H_t (x_t - \bm{\bar{x}}_t)}$, where $\alpha \in (0,1)$ is a line-search parameter.\footnote{Since DDP is a second-order method and potentially takes  large steps, regularization is required to prevent the blow-up of the value. Therefore, we multiply $k_t$ and $h_t^{(i)}$ by scaling a parameter $\alpha\in(0,1)$ and perform a line-search.  In particular, the line-search parameter alpha is iteratively reduced to improve the total cost.} 
Then, both the control input and the disturbance sample are executed to the system for $t=0,\dots, T-1$ starting from the initial state $x_0$.\footnote{The complexity of a single iteration of our algorithm is bounded by $O\big(T(n_x^3  + n_u^3+ (N +n_w)n_w^2  )\big)$, which is polynomial in state,  input and disturbance dimensions and linear in the time horizon and sample size.}

\begin{algorithm}[t]
\DontPrintSemicolon
\caption{DR-DDP algorithm}
\label{alg:dr_ddp}
\SetKw{Input}{Input:}
\SetKw{to}{to}
\SetKw{and}{and}
\SetKw{return}{return}
\Input{$x_0, \pi_\mathrm{init}, \gamma_\mathrm{init}, T, \lambda$}\;
Apply $(\pi_\mathrm{init},\gamma_\mathrm{init})$ to generate $(\bar{\bm{x}}_\mathrm{nom}, \bar{\bm{u}}_\mathrm{nom}, \bar{\bm{w}}_\mathrm{nom})$\;
\While{not converged}{
    \tcp{Backward Pass}
    $\bm{V}_{T} \gets \ell_f(\bar{\bm{x}}_T),  V_{T,x} \gets \ell_{f,x},  V_{T,xx} \gets \ell_{f,xx}$\;
    \For{$t=T-1$ \to $0$}{
Construct $(\bar{\pi}_t^*, \bar{\gamma}_t^*)$ using~\eqref{aff_cont_pol} and~\eqref{wc_dist_}\;
    Update $\bm{V}_t, V_{t,x}, V_{t,xx}$ according to~\eqref{quad_approx}\;
    }
    \tcp{Forward Pass}
  	 Perform line-search to update $\alpha$\;
    \For{$t=0$ \to $T-1$}{
Compute $u_t = \bar{\bm{u}}_t + \alpha k_t + K_t(x_t - \bar{\bm{x}}_t)$\;
    Sample $w_t \sim  \frac{1}{N}\sum_{i=1}^{N} \delta_{\bm{\bar{w}}_t + \alpha h_t^{(i)} + H_t (x_t - \bm{\bar{x}}_t)}$\;
Execute $u_t$ and $w_t$ to~\eqref{dyn} and observe $x_{t+1}$\;
    }
    $ \bar{\bm{x}}_\mathrm{nom}  \gets x_{0:T}, \bar{\bm{u}}_\mathrm{nom} \gets u_{0:T-1}, \bar{\bm{w}}_\mathrm{nom}  \gets  w_{0:T-1}$
}
\return $(\bar{\pi}^*, \bar{\gamma}^*)$
\end{algorithm}

\section{Numerical Experiments}\label{sec:exp}

In this section, we compare the empirical performance of our DR-DDP method with two baseline algorithms:  \emph{GT-DDP}~\cite{sun2018min}, which uses a minimax approach to consider the worst-case disturbances, and  \emph{box-DDP}~\cite{tassa2014control}, a deterministic DDP algorithm that ignores uncertainties in the controller design but considers box constraints on control inputs.

In our experiments, we choose the penalty parameter $\lambda$ to minimize the cost upper bound~\eqref{guarantee} for $\theta = 0.1$ under the DR-DDP policy pair $(\bar{\pi}^*, \bar{\gamma}^*)$. We estimate the upper bound by conducting 1,000 independent Monte Carlo simulations and computing the Wasserstein distance via a linear program. The optimal penalty parameter is then  found via numerical optimization. All simulations were performed on a PC with a 3.70 GHz Intel Core i7-8700K processor and 32 GB RAM. The source code of our DR-DDP implementation is available online. \footnote{\href{https://github.com/CORE-SNU/DR-DDP}{\tt https://github.com/CORE-SNU/DR-DDP}}

\subsection{Kinematic Car Navigation}

In the first experiment, we consider an autonomous navigation task for a kinematic car in an intersection where a randomly moving obstacle obstructs navigation. 
Consider the following kinematic vehicle model:
\[
x_{t+1} = \begin{bmatrix}x_{t+1}^\mathrm{car} \\ \bm{p}_{t+1}^\mathrm{obs}\end{bmatrix} = \begin{bmatrix}f_\mathrm{car}(x_t^\mathrm{car}, u_t)\\ \bm{p}_{t}^\mathrm{obs} + \Delta \bm{p}_{t}^\mathrm{obs}+ w_t\end{bmatrix}
\]
with system state $x_t \in \mathbb{R}^{5}$ and control input $u_t\in\mathbb{R}^2$. Here, $x_t^\mathrm{car}\in\mathbb{R}^{3}$ represents the car's state evolving according to the differential-drive kinematics $f_\mathrm{car}:\mathbb{R}^{3}\times \mathbb{R}^2\to\mathbb{R}^3$ and consists of the car's center position $\bm{p}$ and its heading angle $\phi$. The control input vector comprises  the velocity and steering angle of the car and has a lower limit of $\underline{\mathbf{u}} = [0, -0.6]^\top$ and an upper limit of $\overline{\mathbf{u}} = [10, 0.6]^\top$. The state component $\bm{p}^\mathrm{obs}_t$ represents the position vector of a random circular obstacle with radius $r_\mathrm{obs}=0.2$.  It is assumed that in each time instance, the obstacle has a pre-specified deterministic motion represented by $\Delta \bm{p}_t^\mathrm{obs} \in \mathbb{R}^{2}$, which is obstructed with a positional disturbance vector $w_t\in\mathbb{R}^{2}$. Each component of the disturbances follows a uniform distribution $\mathcal{U}(-0.001, 0.001)$. Our DR-DDP algorithm uses only $N=10$ samples drawn from the true distribution. The goal is to safely pass the intersection by tracking the given reference trajectory $x^\mathrm{ref}$ and avoiding the obstacle in $T=800$ steps. For this purpose, we design the cost function as 
\[
\begin{split}
\ell_t(x,u) := \, &\|x^\mathrm{car} - x_t^\mathrm{ref}\|_Q^2 + \|u\|_R^2 \\
&+ Q_\mathrm{obs} \exp\left(-0.5\frac{\|\bm{p} - \bm{p}^\mathrm{obs}\|^2}{(r_\mathrm{obs}+r_\mathrm{safe})^2}\right),
\end{split}
\]
where the last term is a soft constraint for avoiding the obstacle with a safe margin of $r_\mathrm{safe}=0.2$. The weights are chosen as $Q = 10I, R=0.1I$ and $Q_\mathrm{obs} = 20$. The terminal cost is similar to the running cost with no control cost.  The penalty parameter is set to $\lambda=9000$, which is found as the minimizer of the upper bound in~\eqref{guarantee}. 

\begin{figure}[t]
    \centering
    \includegraphics[width=0.8\linewidth]{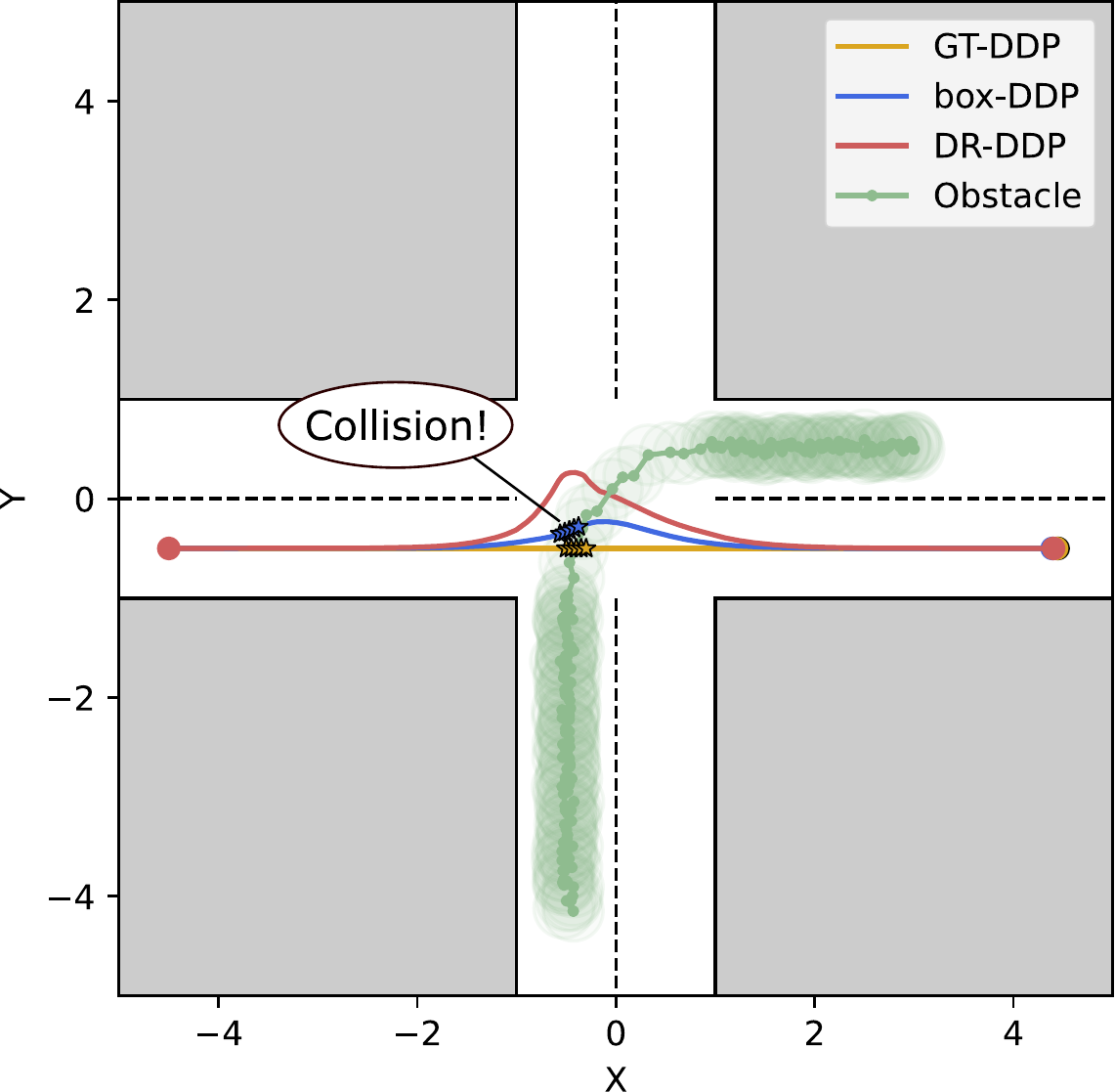}
    \caption{Trajectories of the kinematic car, controlled by GT-DDP, box-DDP, and DR-DDP, in the presence of a randomly moving obstacle. Star marks represent collisions.}
    \label{fig:car}
        \vspace{-0.15in}
\end{figure}

Fig.~\ref{fig:car} shows the trajectories of the kinematic car for a single realization of the disturbances, while Table~\ref{tab:res} summarizes the computation time required for each algorithm. Only DR-DDP  successfully avoids the obstacle and accomplishes the task. Even though box-DDP drives the car away from the reference path, it results in a collision due to its inability to handle uncertainties. Meanwhile, GT-DDP is overly conservative, inducing extremely high costs, thus resulting in a collision.   Despite the distinct behaviors exhibited by the algorithms, the average total times for DR-DDP and GT-DDP are quite similar (less than $25\, sec.$), indicating their comparable computational efficiency.  To validate our results, we conducted 1,000 independent simulation runs to measure the \emph{out-of-sample performance} of each method, which are reported in Table~\ref{tab:res}.\footnote{The out-of-sample performance of the controller 
is defined as $\mathbb{E}^{w_t\sim\mathbb{Q}_t^\mathrm{true}}[\ell_f(x_T) + \sum_{t=0}^{T-1} \ell(x_t, \pi_t^*(x_t))]$, which is evaluated  using 10,000 disturbance samples drawn from the true distribution $\mathbb{Q}_t^\mathrm{true}$ and averaged over 200 simulations. It represents the expected total cost under a new disturbance sample generated according to the true disturbance distribution $\mathbb{Q}_t^\mathrm{true}$ independent of the sample dataset used in DR-DDP.}
The proposed DR-DDP algorithm achieves an out-of-sample cost as low as 172.805, while both box-DDP and GT-DDP demonstrate worse out-of-sample performance  costs of 235.248 and 201.726, respectively. These findings demonstrate the effectiveness of our algorithm in addressing distributional uncertainties in nonlinear stochastic systems.

\begin{table}[t]
\caption{Out-of-sample cost, total computation time, and average computation time per iteration for all algorithms computed over 1,000 simulations.}
\centering
\setlength{\tabcolsep}{0.2em} 
\begin{tabular}{>{\raggedright}m{2cm} >{\centering}m{2cm} >{\centering}m{2cm}  >{\centering\arraybackslash}m{2cm}}
\toprule
& DR-DDP & GT-DDP & box-DDP \\
\midrule
\textbf{Out-of-sample cost} & $172.805$  & $201.726$ & $235.248$\\
\midrule
\textbf{Total comp. time (sec.)} & $24.696$ & $23.587$ & $8.532$\\
\midrule
\textbf{Comp. time per. iter. (sec.)} & $0.215$ & $0.342$  & $0.086$  \\
\bottomrule
\end{tabular}
\label{tab:res}
\vspace{-0.1in}
\end{table}

\subsection{Synchronization of Coupled Oscillators}

In the second experiment, we demonstrate the scalability of our algorithm through a synchronization problem with $L$ coupled noisy oscillators using the following discrete-time Kuramoto model~\cite{kuramoto2012chemical}:
\[
\theta_{t+1}^{(i)} = \theta_{t}^{(i)} + \Delta t \Big(\omega_i + \mathcal{K}  u_t \sum_{j=1}^{L}\sin (\theta_t^{(j)} - \theta_t^{(i)})\Big) + w_t^{(i)},
\]
where $i=1,\dots, L$. Here, $x_t = [\theta_t^{(1)},\dots, \theta_t^{(L)}]^\top \in\mathbb{R}^{L}$ is the system state, and $u_t\in\mathbb{R}$ is the control input. For each $i$-th oscillator, $\theta_t^{(i)}$ represents its phase, $\omega^{(i)}$ is its natural frequency, $\mathcal{K}$ is the coupling strength, and $\Delta t=0.03 \, sec.$ is the discretization step. We assume the frequencies $\omega^{(i)}$ and disturbances $w_t^{(i)}$  follow Gaussian distributions $\mathcal{N}(0,0.004)$ and $\mathcal{N}(0.001, 0.001)$, respectively. We aim to synchronize the oscillators within a finite horizon of $T=100$, assuming that only $N=50$ disturbance samples are available. For that, the cost function is designed as
\[
\ell(x_t,u_t):= \sum_{i,j=1}^{L} \sin^2 (\theta_t^{(j)}-\theta_t^{(i)})+ 0.0001 u_t^2,
\]
and the penalty parameter is chosen as $\lambda=10000$ to minimize the upper bound~\eqref{guarantee}.

To assess the scalability of our method, we evaluate the computation time to perform a single iteration of our DR-DDP algorithm depending on the number of oscillators.  The computation times required for our method and two baselines, as well as the corresponding total costs, are reported in Fig.~\ref{fig:osc}. As expected, the computation time increases with the number of oscillators. However, consistent with the theoretical complexity, the computation time grows as a polynomial function of the state dimension, showing the superiority of our method over the DP algorithm. Notably, the computation time required to perform a single iteration of DR-DDP is almost identical to the computation times required by box-DDP and GT-DDP. Furthermore, our DR-DDP algorithm consistently returns the lowest out-of-sample cost for any number of oscillators considered, successfully synchronizing the oscillators despite the disturbances.

\begin{figure}[t]
	 
     \begin{subfigure}[b]{0.478\linewidth}
         \centering
    		\includegraphics[width=\linewidth]{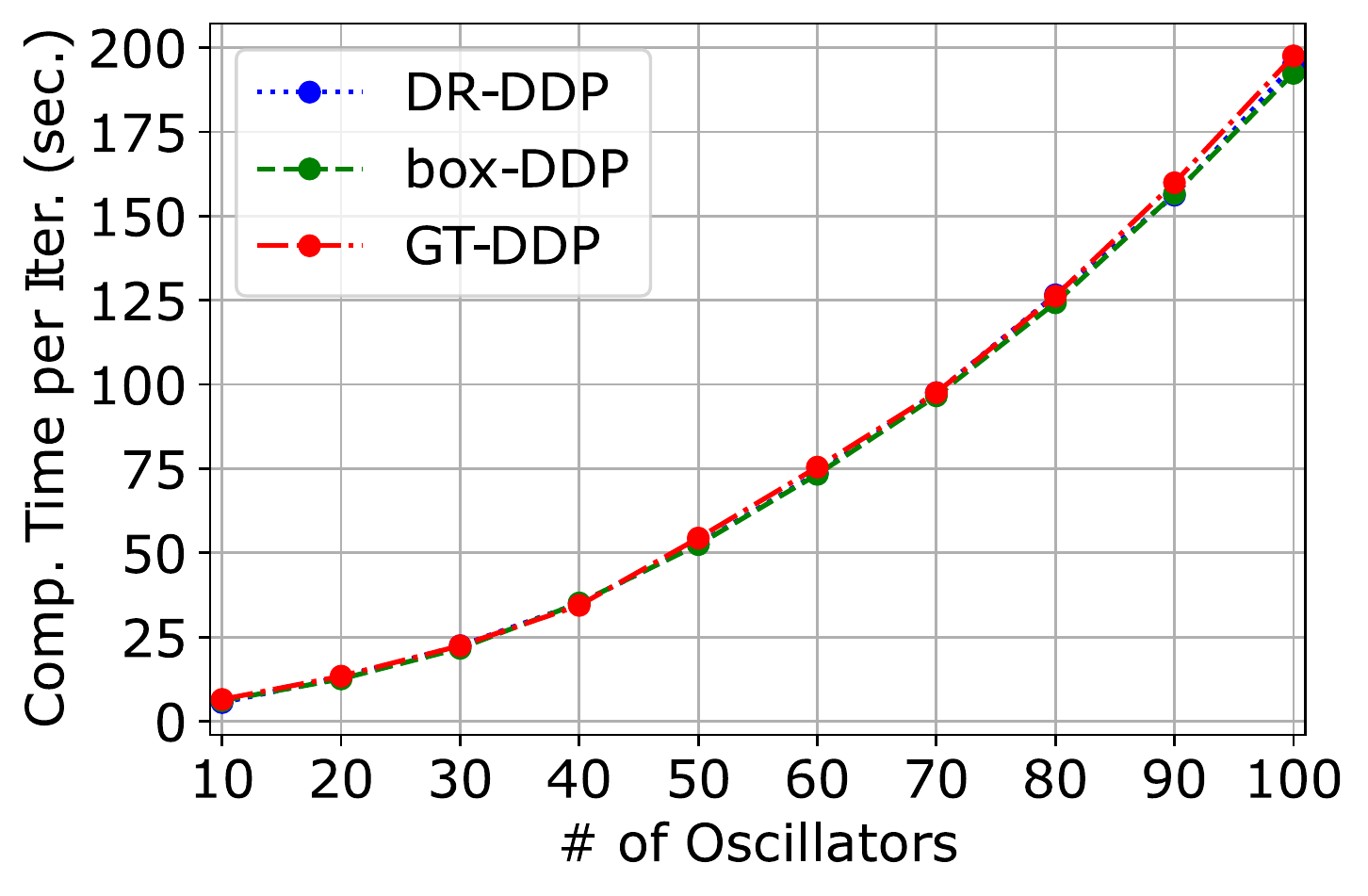}
    		  \caption{ }
    		 \label{fig:comp_time}
     \end{subfigure}
     \hfill
     \begin{subfigure}[b]{0.475\linewidth}
         \centering
         \includegraphics[width=\linewidth]{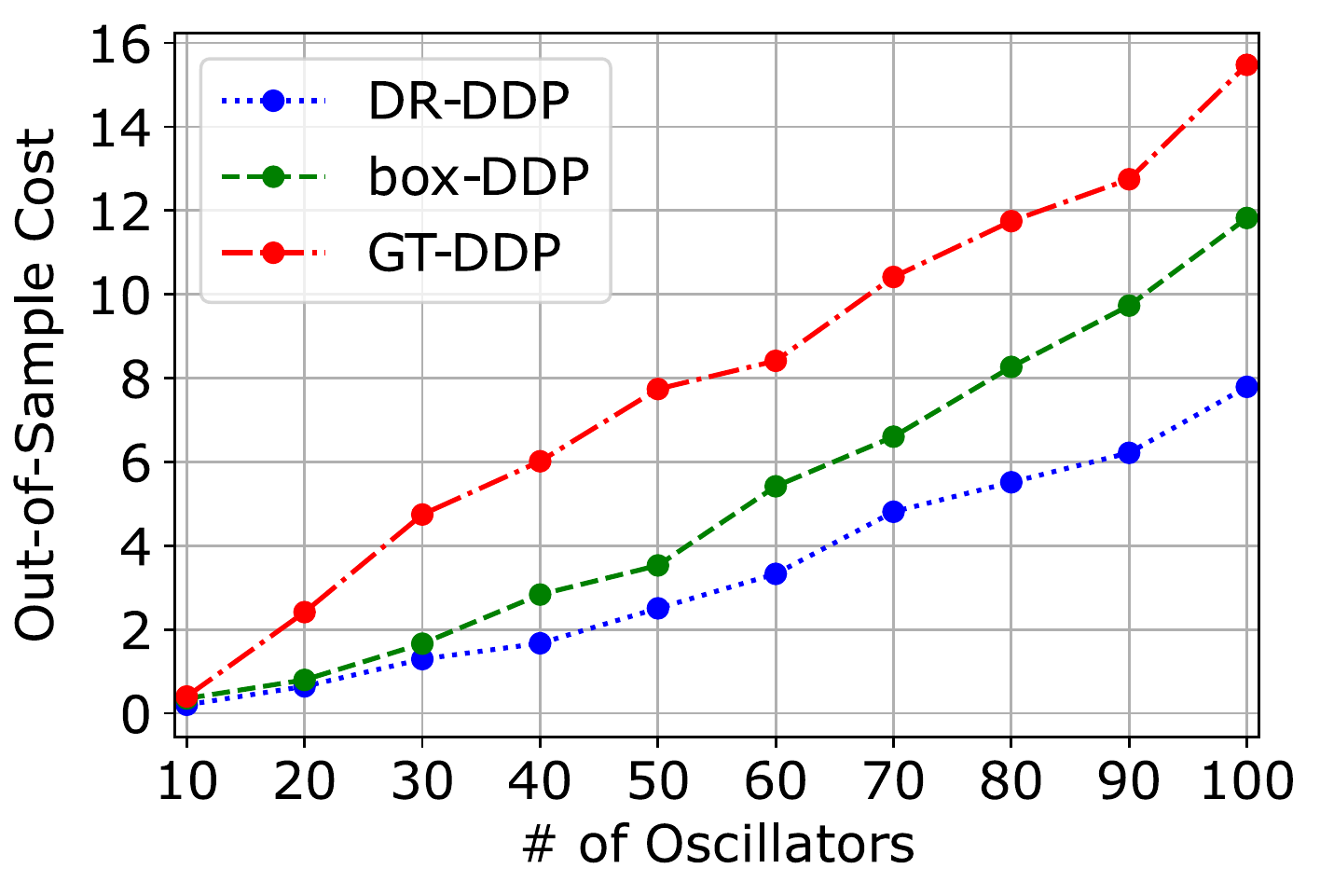}
         \caption{ }
         \label{fig:osp_cost}
     \end{subfigure}
     \caption{(a) Computation time per iteration (in seconds) and (b) out-of-sample cost depending on the number of oscillators calculated with 1,000 simulations.}
     \label{fig:osc}
\end{figure}

\section{Conclusions}
In this work, we have proposed a  practical  DR-DDP algorithm for solving nonlinear stochastic optimal control problems with unknown disturbance distributions. Our approach leverages WDRC to address
limited distributional information. We reformulated the quadratic approximation of  value functions for WDRC using the Kantorovich duality principle and then solved it in a DDP fashion to obtain closed-form expressions of the distributionally robust control and distribution policies in each iteration. Our simulation results demonstrate the superior out-of-sample performance of the proposed method compared to existing DDP methods, as well as its outstanding scalability to high-dimensional state spaces.
In the future, we plan to investigate the theoretical properties of our algorithm, including its convergence rate and performance guarantees.

\appendices

\section{Proof of Proposition~\ref{prop:gc}} \label{app:gc}
\begin{proof}
The proof is based on the arguments used in \cite[Lemma~4.1]{kim2022minimax} for the LQ case. Specifically, fix $\lambda > 0$. For any $\varepsilon > 0$, there exists $\gamma^\varepsilon \in \Gamma_{\mathbb{D}}$ such that
\[
\sup_{\gamma \in \Gamma_{\mathbb{D}}} J(\pi, \gamma) - \epsilon < J(\pi, \gamma^\epsilon).
\]
Since  $\gamma^\epsilon \in \Gamma_{\mathbb{D}}$, it follows that $\gamma_t^\varepsilon(x_t) = \mathbb{P}_t \in \mathbb{D}_t$. Thus,
\begin{equation}\nonumber
\begin{split}
J (\pi, \gamma^\epsilon) &\leq \lambda T \theta^2 + J_\lambda (\pi, \gamma^\epsilon)\\
& \leq \lambda T \theta^2 + \sup_{\gamma \in \Gamma} J_\lambda (\pi, \gamma). 
\end{split}
\end{equation}
Since this inequality holds for any $\varepsilon >0$, we conclude that~\eqref{guarantee} holds.
\end{proof}

\section{Proof of Proposition~\ref{prop:vf}} \label{app:vf}
\begin{proof}
We first note that by the definition of the Wasserstein distance, the inner supremum in~\eqref{vf} is equivalent to
\begin{equation}\label{V_wass}
\begin{split}
& \sup_{\mathbb{P} \in \mathcal{P}(\mathbb{R}^{n_w})}   \mathbb{E}^{w \sim \mathbb{P}}  \bigg[ V_{t+1}(f(\bm{x},\bm{u},w))
- \lambda W_2(\mathbb{P},\mathbb{Q}_t)^2\bigg] \\
= \, &  \sup_{\mathbb{P} \in \mathcal{P}(\mathbb{R}^{n_w})}  \int_{\mathcal{W}} V_{t+1}(f(\bm{x},\bm{u},w)) \mathrm{d} \mathbb{P}(w) \\
&- \lambda \inf_{\substack{\tau \in \mathcal{P}(\mathcal{W}^2):\\ \Pi^1 \tau = \mathbb{P}, \Pi^2 \tau = \mathbb{Q}_t}} \int_{\mathcal{W}^2}\|w-\hat{w}\|^2 \, \mathrm{d}\tau(w, \hat{w})\\
= \, & \sup_{\substack{\tau \in \mathcal{P}(\mathcal{W}^2):\\ \Pi^2 \tau = \mathbb{Q}_t}}  \int_{\mathcal{W}^2} \Big[ V_{t+1}(f(\bm{x},\bm{u},w))  - \lambda \|w-\hat{w}\|^2 \Big] \, \mathrm{d}\tau(w, \hat{w})
\end{split}
\end{equation}

According to the Kantorovich duality principle~\cite{mohajerin2018data, gao2022distributionally},
\[
W_2(\mathbb{P}, \mathbb{Q})^2 = \sup_{\varphi, \psi \in \Phi} \left[\int_\mathcal{W} \varphi(w) \mathrm{d} \mathbb{P}(w) + \int_\mathcal{W} \psi(\hat{w}) \mathrm{d} \mathbb{Q}_t(\hat{w}) \right],
\]
where $\Phi:= \{ (\varphi, \psi) \in L^1(\mathrm{d} w) \times L^2 (\mathrm{d}\hat{w}) \mid \varphi(w) + \psi(\hat{w}) \leq \| w - \hat{w}\|^2, \forall w, \hat{w} \in \mathcal{W}\}$.
Thus, for any $(\varphi, \psi)\in\Phi$, we have that
\[
\psi(\hat{w}) \leq \inf_{\bm{w} \in \mathcal{W}} \| \hat{w} -\bm{w} \|^2 - \varphi(\bm{w})
\]
for each $\hat{w}\in\mathcal{W}$. Consequently, for any $\lambda > 0$, \emph{weak duality} holds for the inner problem as follows:
\begin{equation}\label{weak_dual}
\begin{split}
&\sup_{\substack{\tau \in \mathcal{P}(\mathcal{W}^2):\\ \Pi^2 \tau = \mathbb{Q}_t}}  \int_{\mathcal{W}^2} \Big[ V_{t+1}(f(\bm{x},\bm{u},w))  - \lambda \|\hat{w}-w\|^2 \Big] \, \mathrm{d}\tau(w, \hat{w}) \\
&\leq  \int_{\mathcal{W}} \sup_{\bm{w} \in \mathcal{W}} \Big[ V_{t+1}(f(\bm{x},\bm{u},\bm{w}))  - \lambda \|\hat{w}-\bm{w}\|^2 \Big] \, \mathrm{d}\mathbb{Q}_t(w').
\end{split}
\end{equation}

Using Proposition 1 in~\cite{sinha2017certifying}, we can further show that strong duality holds for the inner problem for any $\lambda > 0$. We conclude the proof by substituting the expression for the inner supremum into~\eqref{vf}.
\end{proof}

\bibliographystyle{IEEEtran}
\bibliography{reference}

\end{document}